\documentclass[preprint,12pt]{elsarticle}

\usepackage{graphicx}

\usepackage{amsfonts,amstext,amsmath,amssymb,amsthm}

\usepackage{rotating}
\usepackage{breakcites}

\usepackage{epstopdf}
\ifpdf
  \DeclareGraphicsExtensions{.eps,.pdf,.png,.jpg}
\else
  \DeclareGraphicsExtensions{.eps}
\fi

\usepackage{mathrsfs}

\usepackage{empheq}

\usepackage{times}
\usepackage{subcaption} 
\usepackage{wrapfig}

\usepackage{makecell}
\usepackage{tabularx, booktabs}
\newcolumntype{Y}{>{\centering\arraybackslash}X}

\usepackage{multirow}
\usepackage{bigstrut}

\usepackage{paralist}

\usepackage{setspace}

\usepackage{bigints}

\usepackage{amsopn}

\makeatletter
\newcommand\makebig[2]{%
  \@xp\newcommand\@xp*\csname#1\endcsname{\bBigg@{#2}}%
  \@xp\newcommand\@xp*\csname#1l\endcsname{\@xp\mathopen\csname#1\endcsname}%
  \@xp\newcommand\@xp*\csname#1r\endcsname{\@xp\mathclose\csname#1\endcsname}%
}
\makeatother

\makebig{biggg} {3.0}
\makebig{Biggg} {3.5}
\makebig{bigggg}{4.0}
\makebig{Bigggg}{4.5}

\usepackage{hyperxmp}

\usepackage[%
    bookmarks=false,         
    unicode=false,          
    pdftoolbar=true,        
    pdfmenubar=true,        
    pdffitwindow=false,     
    pdfstartview={FitH},    
    pdftitle={Analytic Evaluation of the Fractional Moments for the Quasi-Stationary Distribution of the Shiryaev Martingale on an Interval},    
    pdfauthor={Kexuan Li, Aleksey S. Polunchenko and Andrey Pepelyshev},     
    pdfsubject={Analytic Evaluation of the Fractional Moments for the Quasi-Stationary Distribution of the Shiryaev Martingale on an Interval},   
    pdfcreator={Aleksey S. Polunchenko},   
    pdfproducer={MikTeX}, 
    pdfkeywords={Fractional moments, Markov diffusions, Quasi-stationarity, Shiryaev process, Stochastic processes}, 
    pdfnewwindow=true,      
    colorlinks=false,       
    linkcolor=red,          
    citecolor=green,        
    filecolor=magenta,      
    urlcolor=cyan,           
    breaklinks=true
]{hyperref}

\usepackage[utf8]{inputenc}
\usepackage[T1]{fontenc}
\usepackage{textcomp}

\biboptions{authoryear,compress,round}

\newcommand{\pref}[1]{(\ref{#1})}

\newcommand{\ignore}[1]{}

\renewcommand{\Pr}{\mathbb{P}} 
\DeclareMathOperator{\EV}{\mathbb{E}} 

\newcommand{\abs}[1]{\left\vert#1\right\vert}

\renewcommand{\le}{\leqslant} 
\renewcommand{\ge}{\geqslant}

\renewcommand{\Re}{\operatorname{Re}}

\DeclareMathOperator{\ONE}{\mathchoice{\rm 1\mskip-4.2mu l}{\rm 1\mskip-4.2mu l}{\rm 1\mskip-4.6mu l}{\rm 1\mskip-5.2mu l}}
\newcommand{\indicator}[1]{\ONE_{\left\{#1\right\}}}

\theoremstyle{plain} 

\newtheorem{lemma}{Lemma}[section]
\newtheorem{remark}{Remark}[section]

\graphicspath{{./gfx/}}

\journal{Communications in Statistics---Simulation and Computation}

\begin{document}

\begin{frontmatter}



\title{\large{\bf\uppercase{Analytic Evaluation of the Fractional Moments for the Quasi-Stationary Distribution of the Shiryaev Martingale on an Interval}}}


\author[bu]{Kexuan\ Li}
\ead{kli@math.binghamton.edu}
\ead[url]{http://www2.math.binghamton.edu/p/people/grads/kli/start}
\author[bu]{Aleksey\ S.\ Polunchenko\corref{cor-author}}
\ead{aleksey@binghamton.edu}
\address[bu]{Department of Mathematical Sciences, State University of New York (SUNY) at Binghamton\\
Binghamtom, NY 13902--6000, USA}
\ead[url]{http://people.math.binghamton.edu/aleksey}
\author[cardiff]{Andrey\ Pepelyshev}
\ead{pepelyshevan@cardiff.ac.uk}
\address[cardiff]{School of Mathematics, Cardiff University, Cardiff, CF24 4AG, UK\\
and\\
St. Petersburg State University, St. Petersburg, Russia, 199034}
\cortext[cor-author]{Address correspondence to A.S.\ Polunchenko, Department of Mathematical Sciences, State University of New York (SUNY) at Binghamton, Binghamton, NY 13902--6000, USA; Tel: +1 (607) 777-6906; Fax: +1 (607) 777-2450; Email:~\href{mailto:aleksey@binghamton.edu}{aleksey@binghamton.edu}}

\begin{abstract}
We consider the quasi-stationary distribution of the classical Shiryaev diffusion restricted to the interval $[0,A]$ with absorption at a fixed $A>0$. We derive analytically a closed-form formula for the distribution's fractional moment of an {\em arbitrary} given order $s\in\mathbb{R}$; the formula is consistent with that previously found by~\cite{Polunchenko+Pepelyshev:SP2018} for the case of $s\in\mathbb{N}$. We also show by virtue of the formula that, if $s<1$, then the $s$-th fractional moment of the quasi-stationary distribution becomes that of the exponential distribution (with mean $1/2$) in the limit as $A\to+\infty$; the limiting exponential distribution is the stationary distribution of the reciprocal of the Shiryaev diffusion.
\end{abstract}


\end{frontmatter}

\section{Introduction}
\label{sec:intro}

This paper is an extension of the recent work of~\cite{Polunchenko+Pepelyshev:SP2018} devoted to explicit analytic evaluation of \begin{inparaenum}[\itshape(a)]\item the $n$-th moment, with $n\in\mathbb{N}$; and \item the Laplace transform \end{inparaenum} of the quasi-stationary distribution of the classical Shiryaev diffusion restricted to an interval; the Shiryaev process in question and its quasi-stationary distribution are formally introduced below. The principal contribution of the present work is a closed-form formula for the quasi-stationary distribution's {\em fractional} moment of {\em any} given order $s\in\mathbb{R}$. Specifically, we consider the Shiryaev process defined as the solution $(R_{t}^{r})_{t\ge0}$ of the stochastic differential equation
\begin{equation}\label{eq:Rt_r-def}
dR_{t}^{r}
=
dt+R_{t}^{r}dB_t
\;\;
\text{with}
\;\;
R_{0}^{r}\coloneqq r\ge0
\;\;
\text{fixed},
\end{equation}
where $(B_t)_{t\ge0}$ is standard Brownian motion in the sense that $\EV[dB_t]=0$, $\EV[(dB_t)^2]=dt$, and $B_0=0$; the initial value $r$ is often referred to as the process' headstart. The time-homogeneous Markov process $(R_{t}^{r})_{t\ge0}$ is an important particular version of the so-called generalized Shiryaev process. The latter has been first arrived at and studied by Prof. A.N. Shiryaev---hence, the name---in his fundamental work on quickest change-point detection where the process is typically referred to as the Shiryaev--Roberts process; see, e.g.,~\cite{Shiryaev:SMD61,Shiryaev:TPA63}. While the interest to the Shiryaev process in the context of quickest change-point detection has always been strong (see, e.g.,~\citealt{Pollak+Siegmund:B85,Shiryaev:Bachelier2002,Feinberg+Shiryaev:SD2006,Burnaev+etal:TPA2009,Polunchenko:SA2016,Polunchenko:SA2017a,Polunchenko:SA2017b,Polunchenko:TPA2017}), the process has received a great deal of attention in other areas as well, notably mathematical finance (see, e.g.,~\citealt{Geman+Yor:MF1993,Donati-Martin+etal:RMI2001,Linetsky:OR2004}) and mathematical physics (see, e.g.,~\citealt{Monthus+Comtet:JPhIF1994,Comtet+Monthus:JPhA1996}). It has also been studied in the literature on general stochastic processes (see, e.g.,~\citealt{Wong:SPMPE1964,Yor:AAP1992,Donati-Martin+etal:RMI2001,Dufresne:AAP2001,Schroder:AAP2003,Peskir:Shiryaev2006,Polunchenko+Sokolov:MCAP2016,Polunchenko+etal:TPA2018}).

The particular version $(R_{t}^{r})_{t\ge0}$ of the generalized Shiryaev process governed by equation~\eqref{eq:Rt_r-def} is of special importance and interest because it is the {\em only} version with probabilistically nontrivial behavior in the limit as $t\to+\infty$. This behavior is manifested through a certain invariant probability measure attained by the process (in the limit as $t\to+\infty$) in spite of the distinct martingale property $\EV[R_{t}^{r}-r-t]=0$ for all $t\ge0$ and $r\ge0$. Moreover, the process is convergent (as $t\to+\infty$) regardless of whether the state space is \begin{inparaenum}[(I)]\item the entire half-line $[0,+\infty)$ with no absorption on the interior; or\label{lst:GSR-QSD-case-I} \item the interval $[0,A]$ with absorption at a given level $A>0$; or\label{lst:GSR-QSD-case-II} \item the shortened half-line $[A,+\infty)$ also with absorption at $A>0$ given.\label{lst:GSR-QSD-case-III} \end{inparaenum} The case of a negative initial value $r$ was touched upon by~\cite{Peskir:Shiryaev2006} who, in particular, showed that, if $r\in(-\infty,0]$, then $(R_{t}^{r})_{t\ge0}$ has no invariant probability measure on the negative half-line, and is simply bound (with unit probability) to eventually exit the region through the upper (exit) boundary. Cases~\pref{lst:GSR-QSD-case-I},~\pref{lst:GSR-QSD-case-II}, and~\pref{lst:GSR-QSD-case-III} have all been considered in the literature. Brief surveys of the corresponding results and relevant references were recently offered, e.g., by~\cite{Polunchenko+etal:TPA2018} and by~\cite{Polunchenko+Pepelyshev:SP2018}.

This work's focus is on case~\pref{lst:GSR-QSD-case-II}. This case is of importance in quickest change-point detection, and in this context, it was investigated, e.g., by~\cite{Pollak+Siegmund:B85},~\cite{Burnaev+etal:TPA2009} and by~\cite{Polunchenko:SA2017a}. See also, e.g.,~\cite{Pollak+Siegmund:JAP1996},~\cite{Linetsky:OR2004} and~\cite[Section~7.8.2]{Collet+etal:Book2013}. For a fixed $A>0$ and a given $t>0$, the distribution of $R_{t}^{r}$ in this case is conditional on no absorption prior to $t$. This distribution was derived explicitly, e.g., by~\cite{Linetsky:OR2004} and by~\cite{Polunchenko:SA2016}. The limit of this distribution as $t\to+\infty$ is known as the {\em quasi}-stationary distribution. Formally, consider the stopping time
\begin{equation*}
\mathcal{S}_{A}^{r}
\coloneqq
\inf\{t\ge0\colon R_{t}^{r}=A\}
\;\;
\text{such that}
\;\;
\inf\{\varnothing\}=+\infty,
\end{equation*}
where $R_{0}^{r}\coloneqq r\ge0$ and $A>0$ are fixed; note that $\Pr(\mathcal{S}_{A}^{r}<+\infty)=1$. The quasi-stationary distribution of $(R_{t}^{r})_{t\ge0}$ is defined as
\begin{equation}\label{eq:QSD-def}
Q_{A}(x)
\coloneqq
\lim_{t\to+\infty}\Pr(R_{t}^{r}\le x|\mathcal{S}_{A}^{r}>t)
\;\;
\text{and}
\;\;
q_{A}(x)
\coloneqq
\dfrac{d}{dx}Q_{A}(x),
\end{equation}
and it does not depend on $r\in[0,A)$. The existence of this distribution was formally established, e.g., by~\cite{Pollak+Siegmund:B85}, although one can also infer the same conclusion, e.g., from the seminal work of~\cite{Mandl:CMJ1961}. Moreover, analytic closed-form formulae for $Q_{A}(x)$ and $q_{A}(x)$ were recently obtained by~\cite{Polunchenko:SA2017a}, apparently for the first time in the literature; see formulae~\eqref{eq:QSD-pdf-answer} and~\eqref{eq:QSD-cdf-answer} in Section~\ref{sec:preliminaries} below.

If $A\to+\infty$, then case~\pref{lst:GSR-QSD-case-II} becomes case~\pref{lst:GSR-QSD-case-III}, but the process $(R_{t}^{r})_{t\ge0}$ still has a nontrivial probabilistic behavior in the limit as $t\to+\infty$. This behavior is characterized by the invariant probability measure known as the {\em stationary} distribution. Formally, for the Shiryaev process~\eqref{eq:Rt_r-def} the latter is defined as
\begin{equation}\label{eq:SR-StDist-def}
H(x)
\coloneqq
\lim_{t\to+\infty}\Pr(R_{t}^{r}\le x)
\;\;
\text{and}
\;\;
h(x)
\coloneqq
\dfrac{d}{dx}H(x),
\end{equation}
and it, too, is independent of $r\in[0,+\infty)$. This distribution has already been found, e.g., by~\cite{Shiryaev:SMD61,Shiryaev:TPA63}, \cite{Pollak+Siegmund:B85}, \cite{Feinberg+Shiryaev:SD2006}, \cite{Burnaev+etal:TPA2009}, and \cite{Polunchenko+Sokolov:MCAP2016}, to be the momentless (no moments of orders one and higher) distribution
\begin{equation}\label{eq:SR-StDist-answer}
H(x)
=
e^{-\tfrac{2}{x}}\indicator{x\ge0}
\;\;
\text{and}
\;\;
h(x)
=
\dfrac{2}{x^2}\,e^{-\tfrac{2}{x}}\indicator{x\ge0},
\end{equation}
which is an extreme-value Fr\'{e}chet-type distribution, and a particular case of the inverse (reciprocal) Gamma distribution. See also, e.g.,~\cite{Linetsky:OR2004} and~\cite{Avram+etal:MPRF2013}. It is immediate to see from~\eqref{eq:SR-StDist-answer} that the stationary distribution of the reciprocal of the Shiryaev process is exponential with mean $1/2$. The fact that $Q_{A}(x)$ converges to $H(x)$ as $A\to+\infty$ was established, e.g., by~\cite{Pollak+Siegmund:B85,Pollak+Siegmund:JAP1996}; the convergence is from above, and is pointwise, at every $x\in[0,+\infty)$, i.e., at all continuity points of $H(x)$. The $Q$-to-$H$ convergence was recently investigated further by~\cite{Li+Polunchenko:SA2019} who used the explicit formula obtained by~\cite{Polunchenko:SA2017a} for $Q_{A}(x)$ to demonstrate that the uniform (in $x$) rate of convergence of $Q_{A}(x)$ down to $H(x)$ as $A\to+\infty$ is no slower than $O(\log(A)/A)$.

While the stationary distribution~\eqref{eq:SR-StDist-def}--\eqref{eq:SR-StDist-answer} is momentless (no moments of orders one and higher, but moments of orders $s<1$ do exist), the quasi-stationary distribution~\eqref{eq:QSD-def} is not: its entire (positive-integer-order) moment series was recently evaluated explicitly by~\cite{Polunchenko+Pepelyshev:SP2018} through the use of the closed-form formulae for $Q_{A}(x)$ and $q_{A}(x)$ obtained earlier by~\cite{Polunchenko:SA2017a}. The contribution of this work is an explicit formula (in different equivalent forms) for the quasi-stationary distribution's {\em fractional} moment of {\em any} given order $s\in\mathbb{R}$; the existence of the $s$-th order fractional moment is shown formally as well. All this is done in Section~\ref{sec:moment-formulae}, the core part of the paper. Since the $s$-th fractional moment formula and its derivation involve certain special functions, the centerpiece of the paper is prefaced with Section~\ref{sec:nomenclature} which introduces the relevant special functions. Section~\ref{sec:preliminaries} conveniently summarizes the essential earlier results on the quasi-stationary distribution~\eqref{eq:QSD-def} of the Shiryaev process~\eqref{eq:Rt_r-def}. Lastly, concluding remarks are provided in Section~\ref{sec:conclusion}.

\section{Notation and nomenclature}
\label{sec:nomenclature}

For convenience, standard mathematical notation will be used throughout the sequel. This includes the small array of special functions we shall deal with in the sections to follow. These functions, in their most common notation, are:
\begin{enumerate}
    \setlength{\itemsep}{10pt}
    \setlength{\parskip}{0pt}
    \setlength{\parsep}{0pt}
    \item The Gamma function $\Gamma(z)$, $z\in\mathbb{C}$, frequently also referred to as the extension of the factorial to complex numbers, due to the property $\Gamma(n)=(n-1)!$ exhibited for $n\in\mathbb{N}$. See, e.g.,~\cite[Chapter~1]{Bateman+Erdelyi:Book1953v1}.
    \item The digamma function $\psi(z)$, $z\in\mathbb{C}$, defined as the logarithmic derivative of the Gamma function:
    \begin{equation}\label{eq:digamma-fun-def}
    \psi(z)
    \coloneqq
    \dfrac{d}{dz}\log\Gamma(z)
    =
    \dfrac{\Gamma'(z)}{\Gamma(z)};
    \end{equation}
    see, e.g.,~\cite[Section~1.7]{Bateman+Erdelyi:Book1953v1}.
    \item The Pochhammer symbol, or the rising factorial, often notated as $(z)_n$ and defined for $z\in\mathbb{C}$ and $n\in\mathbb{N}\cup\{0\}$ as
    \begin{equation*}
        (z)_{n}
        \coloneqq
        \begin{cases}
        1,&\text{for $n=0$};\\
        z(z+1)\cdots(z+n-1),&\text{for $n\in\mathbb{N}$},
        \end{cases}
    \end{equation*}
    and it is of note that $(1)_n=n!$ for any $n\in\mathbb{N}\cup\{0\}$. See, e.g.,~\cite[pp.~16--18]{Srivastava+Karlsson:Book1985}. Also, observe that
    \begin{equation}\label{eq:PochhammerGamma-def}
    (z)_{n}
    =
    \dfrac{\Gamma(z+n)}{\Gamma(z)}
    \;
    \text{for}
    \;
    n\in\mathbb{N}\cup\{0\}
    \;
    \text{and}
    \;
    z\in\mathbb{C}\setminus\{0,-1,-2,\ldots\},
    \end{equation}
    and if $z$ is a negative integer or zero, i.e., if $z=-k$ and $k\in\mathbb{N}\cup\{0\}$, then
    \begin{equation}\label{eq:Pochhammer-negint}
    (-k)_{n}
    =
    \begin{cases}
    \dfrac{(-1)^{n}\,k!}{(k-n)!},&\text{for $n=0,1,\ldots,k$};\\[2mm]
    0,&\text{for $n=k+1,k+2,\ldots$};
    \end{cases}
    \end{equation}
    cf.~\cite[p.~16--17]{Srivastava+Karlsson:Book1985}.
    \item The generalized hypergeometric function, denoted as ${}_{p}F_{q}[z]$, and defined via the power series
    \begin{equation}\label{eq:pFq-function-def}
        {}_{p}F_{q}
        \left[
            \setlength{\arraycolsep}{0pt}
            \setlength{\extrarowheight}{2pt}
            \begin{array}{@{} c@{{}{}} @{}}
            a_1,a_2,\ldots,a_p
            \\[1ex]
            b_1,b_2,\ldots,b_q
            \\[2pt]
            \end{array}
            \;\middle|\;
            z
        \right]
        \coloneqq
        \sum_{n=0}^{\infty}\dfrac{(a_1)_n\,(a_2)_n\,\ldots\,(a_p)_n}{(b_1)_n\,(b_2)_n\,\ldots\,(b_q)_n}\,\dfrac{z^{n}}{n!},
    \end{equation}
    where $p,q\in\mathbb{N}$; it is to be understood here that the argument $z$ and the numerator and denominator parameters, $a_1,\ldots,a_p$ and $b_1,\ldots,b_q$, are such that the series on the right of~\eqref{eq:pFq-function-def} is convergent. At the very least, in view of~\eqref{eq:Pochhammer-negint}, no $b_i$, $i=1,2,\ldots,q$, should be a negative integer or zero. For more details on the convergence conditions, see, e.g.,~\cite[Chapter~4]{Bateman+Erdelyi:Book1953v1} and~\cite[p.~20]{Srivastava+Karlsson:Book1985}. The convergence question is trivial when one (or more) of the numeratorial parameters, say $a_i$ for some $i=1,2,\ldots,p$, is a negative integer or zero: in this case, in view of~\eqref{eq:Pochhammer-negint}, the power series on the right of~\eqref{eq:pFq-function-def} terminates turning the ${}_{p}F_{q}[z]$ function into a polynomial in $z$ of degree $(-a_i)\in\mathbb{N}\cup\{0\}$. Two special cases of the ${}_{p}F_{q}[z]$ function will arise in the sequel: one with $p=q=1$ and one with $p=q=2$. The ${}_{1}F_{1}[z]$ function corresponding to the former case is known as the Kummer~\citeyearpar{Kummer:FAM1836} function. For more information on the Kummer function and its properties, see, e.g.,~\cite{Slater:Book1960} and~\cite[Chapter~I, Section~1]{Buchholz:Book1969}.
    \item The Whittaker $M$ and $W$ functions, traditionally denoted, respectively, as $M_{a,b}(z)$ and $W_{a,b}(z)$, where $a,b,z\in\mathbb{C}$; the Whittaker $M$ function is undefined when $-2b\in\mathbb{N}$, but can be regularized. These functions were introduced by~\cite{Whittaker:BAMS1904} as the fundamental solutions to the Whittaker differential equation. See, e.g.,~\cite{Slater:Book1960} and~\cite{Buchholz:Book1969}.
\end{enumerate}

\section{Preliminary background on the quasi-stationary distribution}
\label{sec:preliminaries}

As was mentioned in the introduction, the quasi-stationary distribution defined in~\eqref{eq:QSD-def} was recently expressed analytically by~\cite{Polunchenko:SA2016}; the expression involves the Whittaker $W$ function. Specifically, it can be deduced from~\cite[Theorem~3.1]{Polunchenko:SA2016} that if $A>0$ is fixed and $\lambda\equiv\lambda_A>0$ is the smallest (positive) solution of the equation
\begin{equation}\label{eq:lambda-eqn}
W_{1,\tfrac{1}{2}\xi(\lambda)}\left(\dfrac{2}{A}\right)
=
0,
\end{equation}
where
\begin{equation}\label{eq:xi-def}
\xi(\lambda)
\coloneqq
\sqrt{1-8\lambda}
\;\;
\text{so that}
\;\;
\lambda
=
\dfrac{1}{8}\left(1-\big[\xi(\lambda)\big]^2\right),
\end{equation}
then the quasi-stationary probability density function (pdf) is given by
\begin{equation}\label{eq:QSD-pdf-answer}
q_A(x)
=
\dfrac{e^{-\tfrac{1}{x}}\,\dfrac{1}{x}\,W_{1,\tfrac{1}{2}\xi(\lambda)}\left(\dfrac{2}{x}\right)}{e^{-\tfrac{1}{A}}\,W_{0,\tfrac{1}{2}\xi(\lambda)}\left(\dfrac{2}{A}\right)}\indicator{x\in[0,A]},
\end{equation}
and the respective cumulative distribution function (cdf) is given by
\begin{equation}\label{eq:QSD-cdf-answer}
Q_A(x)
=
\begin{cases}
1,&\;\text{if $x\ge A$;}\\[2mm]
\dfrac{e^{-\tfrac{1}{x}}\,W_{0,\tfrac{1}{2}\xi(\lambda)}\left(\dfrac{2}{x}\right)}{e^{-\tfrac{1}{A}}\,W_{0,\tfrac{1}{2}\xi(\lambda)}\left(\dfrac{2}{A}\right)},&\;\text{if $x\in[0,A)$;}\\[8mm]
0,&\;\text{otherwise},
\end{cases}
\end{equation}
and $q_A(x)$ and $Q_A(x)$ are each a smooth function of $x$ and $A$; observe from~\eqref{eq:lambda-eqn}, \eqref{eq:xi-def}, and~\eqref{eq:QSD-pdf-answer} that $q_A(A)=0$. The smoothness of $q_A(x)$ and $Q_A(x)$ is due to certain analytic properties of the Whittaker $W$ function present on the right of formulae~\eqref{eq:QSD-pdf-answer} and~\eqref{eq:QSD-cdf-answer}. These formulae stem from the solution of a certain Sturm--Liouville problem, and $\lambda$ is the smallest positive eigenvalue of the corresponding Sturm--Liouville operator; it should be noted that \cite{Polunchenko:SA2016} considered the {\em negated} Sturm--Liouville operator, causing $\lambda$ to be its largest {\em negative} eigenvalue.
\begin{remark}\label{rem:xi-symmetry}
The definition~\eqref{eq:xi-def} of $\xi(\lambda)$ can actually be changed to $\xi(\lambda)\coloneqq -\sqrt{1-8\lambda}$ with no effect whatsoever on either equation~\eqref{eq:lambda-eqn}, or formulae~\eqref{eq:QSD-pdf-answer} and~\eqref{eq:QSD-cdf-answer}, i.e., all three are invariant with respect to the sign of $\xi(\lambda)$. This was previously pointed out by~\cite{Polunchenko:SA2017a}, and the reason for this $\xi(\lambda)$-symmetry is because equation~\eqref{eq:lambda-eqn} and formulae~\eqref{eq:QSD-pdf-answer} and~\eqref{eq:QSD-cdf-answer} each have $\xi(\lambda)$ present only as (double) the second index of the corresponding Whittaker $W$ function or functions involved, and the Whittaker $W$ function in general is known (see, e.g.,~\citealt[Identity~(19),~p.~19]{Buchholz:Book1969}) to be an even function of its second index, i.e., $W_{a,b}(z)=W_{a,-b}(z)$.
\end{remark}

It is evident that equation~\eqref{eq:lambda-eqn} is a key ingredient of formulae~\eqref{eq:QSD-pdf-answer} and~\eqref{eq:QSD-cdf-answer}, and consequently, of all of the characteristics of the quasi-stationary distribution as well. As a transcendental equation, it can only be solved numerically, although to within any desired accuracy; numerical analyses of equation~\eqref{eq:lambda-eqn} were previously carried out, e.g., by~\cite{Linetsky:OR2004} and by~\cite{Polunchenko:SA2016,Polunchenko:SA2017a,Polunchenko:SA2017b}. Yet, it is known (see, e.g.,~\citealt{Linetsky:OR2004} and~\citealt{Polunchenko:SA2016}) that for any fixed $A>0$, equation~\eqref{eq:lambda-eqn} has countably many simple solutions $0<\lambda_1<\lambda_2<\lambda_3<\cdots$, such that $\lim_{k\to+\infty}\lambda_k=+\infty$. All of them depend on $A$, but since we are interested only in the smallest one, we shall use either the ``short'' notation $\lambda$, or the more explicit notation $\lambda_A$ to emphasize the dependence on $A$. Also, it can be concluded from~\cite[p.~136~and~Lemma~3.3]{Polunchenko:SA2016} that $\lambda_A$ is a monotonically decreasing function of $A$ such that $\lim_{A\to+\infty}\lambda_A=0$. More specifically, certain properties of the moments of the quasi-stationary distribution~\eqref{eq:QSD-pdf-answer}--\eqref{eq:QSD-cdf-answer} necessitate that
\begin{equation}\label{eq:lambda-dbl-ineq-asymp}
\dfrac{1}{A}+\dfrac{1}{A(A+1)}
<
\lambda_{A}
<
\dfrac{1}{A}+\dfrac{1+\sqrt{4A+1}}{2A^2}
\;
\text{for any}
\;
A>0,
\;
\text{so that}
\;
\lambda_{A}
=
\dfrac{1}{A}+O\left(\dfrac{1}{\sqrt{A^{3}}}\right);
\end{equation}
cf.~\cite[p.~136]{Polunchenko:SA2017a} and~\cite[p.~1360]{Polunchenko+Pepelyshev:SP2018}.

\section{The fractional moment formulae}
\label{sec:moment-formulae}

Let $X$ be a random variable sampled from the quasi-stationary distribution given by~\eqref{eq:QSD-pdf-answer}--\eqref{eq:QSD-cdf-answer}. Fix $s\in\mathbb{R}$ and let $\mathfrak{M}_{s}\coloneqq\EV[X^{s}]$ denote the $s$-th fractional moment of $X$. It is apparent that $\mathfrak{M}_{0}\equiv1$ for any $A>0$, but all other $\mathfrak{M}_{s}$'s actually do depend on $A$, and, if necessary, we shall stress this dependence via the notation $\mathfrak{M}_{s}(A)$. The object of this section---the centerpiece of this work---is to find $\mathfrak{M}_{s}(A)$ analytically for any given $s\in\mathbb{R}$ and $A>0$.

The obvious first question to ask is that of whether the $\mathfrak{M}_{s}$'s actually do exist for each fixed $s\in\mathbb{R}$ and $A>0$. For $s\ge0$ the answer is clearly ``yes''. For $s<0$ the answer is also in the affirmative inasmuch as the definition of $\mathfrak{M}_{s}$ and the explicit formula~\eqref{eq:QSD-pdf-answer} for $q_{A}(x)$ together yield
\begin{equation*}
\mathfrak{M}_{s}
\coloneqq
\bigintsss_{0}^{A}
x^{s}
\,
q_{A}(x)\,dx
=
2^{s}
\bigints_{\tfrac{2}{A}}^{+\infty}
t^{-s}\,
\dfrac{e^{-\tfrac{t}{2}}\,W_{1,\tfrac{\xi(\lambda)}{2}}(t)}{e^{-\tfrac{1}{A}}\,W_{0,\tfrac{\xi(\lambda)}{2}}\left(\dfrac{2}{A}\right)}\,\dfrac{dt}{t},
\end{equation*}
and the latter integral is convergent, for
\begin{equation*}
W_{1,b}(z)
=
z\,e^{-\tfrac{z}{2}}\left[1+O\left(\dfrac{1}{z}\right)\right]
,
\;
\text{as}
\;
z\to+\infty
\;
\text{for any}
\;
b\in\mathbb{C},
\end{equation*}
which is a special case of the more general asymptotic result
\begin{equation*}
W_{a,b}(z)
=
z^{a}\,e^{-\tfrac{z}{2}}\left[1+O\left(\dfrac{1}{z}\right)\right]
,
\;
\text{as}
\;
\abs{z}\to+\infty
\;
\text{for any}
\;
b\in\mathbb{C},
\end{equation*}
established, e.g., in~\cite[Section~16.3]{Whittaker+Watson:Book1927}.

If $s=n\in\mathbb{N}$, then it can be inferred from~\cite[Theorem~3.2]{Polunchenko:SA2017a} that, for $A>0$ fixed, the series $\{\mathfrak{M}_{n}\}_{n\ge0}$ satisfies the recurrence
\begin{equation}\label{eq:QSD-moments-recurrence}
\left(\dfrac{n(n-1)}{2}+\lambda\right)\mathfrak{M}_{n}+n\,\mathfrak{M}_{n-1}
=
\lambda A^{n},
\;\;
n\in\mathbb{N},
\end{equation}
with $\mathfrak{M}_{0}\equiv 1$ for any $A>0$; recall that $\lambda\equiv\lambda_{A}$ and $A$ are interconnected via equation~\eqref{eq:lambda-eqn}. This recurrence was recently solved explicitly by~\cite{Polunchenko+Pepelyshev:SP2018}.
\begin{lemma}[\citealt{Polunchenko+Pepelyshev:SP2018}]
For every $A>0$ fixed, the solution $\{\mathfrak{M}_{n}\}_{n\ge0}$ to the recurrence~\eqref{eq:QSD-moments-recurrence} is given by
\begin{equation}\label{eq:QSD-Mn-answer-2F2}
\mathfrak{M}_{n}
\equiv
\mathfrak{M}_{n}(A)
=
\dfrac{2\lambda\, A^{n}}{n(n-1)+2\lambda}\,
{}_{2}F_{2}
\left[
    \setlength{\arraycolsep}{0pt}
    \setlength{\extrarowheight}{2pt}
    \begin{array}{@{} c@{{}{}} @{}}
    1,-n
    \\[1ex]
    \dfrac{3}{2}+\dfrac{\xi(\lambda)}{2}-n,\dfrac{3}{2}-\dfrac{\xi(\lambda)}{2}-n
    \\[5pt]
    \end{array}
    \;\middle|\;
    \dfrac{2}{A}
\right],
\;
n\in\mathbb{N}\cup\{0\},
\end{equation}
where $\lambda\equiv\lambda_A\;(>0)$ is determined by~\eqref{eq:lambda-eqn} while $\xi(\lambda)$ is defined in~\eqref{eq:xi-def}; recall also that ${}_{2}F_{2}[z]$ denotes the corresponding case of the generalized hypergeometric function~\eqref{eq:pFq-function-def}.
\end{lemma}

It stands to mention that since $n$ in~\eqref{eq:QSD-Mn-answer-2F2} is a nonnegative integer, the right hand side of~\eqref{eq:QSD-Mn-answer-2F2} is actually a polynomial of degree $n$ in $(2/A)$. See~\cite[Formula~(15), p.~1359]{Polunchenko+Pepelyshev:SP2018}. However, while it is essential for the validity of formula~\eqref{eq:QSD-Mn-answer-2F2} that $n$ in it be a nonnegative integer, the recurrence~\eqref{eq:QSD-moments-recurrence} itself holds true even if $n$ is an arbitrary real number. This conclusion can be reached from a more careful analysis of the steps taken by~\cite{Polunchenko:SA2017a} to arrive at a slightly more general version of the recurrence~\eqref{eq:QSD-moments-recurrence}, plus the fact established above that the quasi-stationary distribution~\eqref{eq:QSD-pdf-answer}--\eqref{eq:QSD-cdf-answer} has all fractional moments, of any real order. Thus the fractional moments $\{\mathfrak{M}_{s}\}_{s\in\mathbb{R}}$ follow the recurrence
\begin{equation}\label{eq:QSD-frac-moments-recurrence}
\left(\dfrac{s(s-1)}{2}+\lambda\right)\mathfrak{M}_{s}+s\,\mathfrak{M}_{s-1}
=
\lambda A^{s},
\;\;
s\in\mathbb{R},
\end{equation}
where $\mathfrak{M}_{0}\equiv 1$ for any $A>0$, and our goal of finding $\mathfrak{M}_{s}$ explicitly becomes a matter of solving~\eqref{eq:QSD-frac-moments-recurrence} analytically. To this end, even though~\eqref{eq:QSD-frac-moments-recurrence} was obtained from~\eqref{eq:QSD-moments-recurrence} by (justifiably) replacing $n$ in the latter with $s\in\mathbb{R}$, the same ``trick'' performed on the solution~\eqref{eq:QSD-Mn-answer-2F2} to recurrence~\eqref{eq:QSD-moments-recurrence} will only yield a {\em particular} solution to~\eqref{eq:QSD-frac-moments-recurrence}, but not the {\em complete} solution. Moreover, formula~\eqref{eq:QSD-Mn-answer-2F2} was arrived at by~\cite{Polunchenko+Pepelyshev:SP2018} as a particular solution to~\eqref{eq:QSD-moments-recurrence} that just happens to satisfy the condition $\mathfrak{M}_{0}\equiv 1$, so there was no need to find the general solution to the homogeneous version of~\eqref{eq:QSD-moments-recurrence}. We shall now fill in this gap, i.e., make the obvious first step toward solving~\eqref{eq:QSD-frac-moments-recurrence} completely.
\begin{lemma}
For every $A>0$ fixed, the {\em general} solution $\{\mathfrak{M}_{s}\}_{s\in\mathbb{R}}$ to the recurrence~\eqref{eq:QSD-frac-moments-recurrence} is given by
\begin{equation}\label{eq:QSD-Ms-gen-answer-2F2}
\begin{split}
\mathfrak{M}_{s}
\equiv
\mathfrak{M}_{s}(A)
&=
\dfrac{2\lambda\, A^{s}}{s(s-1)+2\lambda}
\,
{}_{2}F_{2}
\left[
    \setlength{\arraycolsep}{0pt}
    \setlength{\extrarowheight}{2pt}
    \begin{array}{@{} c@{{}{}} @{}}
    1,-s
    \\[1ex]
    \dfrac{3}{2}+\dfrac{\xi(\lambda)}{2}-s,\dfrac{3}{2}-\dfrac{\xi(\lambda)}{2}-s
    \\[5pt]
    \end{array}
    \;\middle|\;
    \dfrac{2}{A}
\right]
+
\\
&\qquad\qquad\qquad+
C\,\dfrac{2^{s}}{\Gamma(-s)}\,\Gamma\left(\dfrac{1}{2}+\dfrac{\xi(\lambda)}{2}-s\right)\Gamma\left(\dfrac{1}{2}-\dfrac{\xi(\lambda)}{2}-s\right)
,
\;
s\in\mathbb{R},
\end{split}
\end{equation}
where $C\equiv C_{A}$ is an arbitrary constant (independent of $s$ but possibly dependent on $A>0$), and $\lambda\equiv\lambda_A\;(>0)$ is determined by~\eqref{eq:lambda-eqn} while $\xi(\lambda)$ is defined in~\eqref{eq:xi-def}; recall also that ${}_{2}F_{2}[z]$ denotes the corresponding case of the generalized hypergeometric function~\eqref{eq:pFq-function-def}.
\end{lemma}
\begin{proof}
The right-hand side of formula~\eqref{eq:QSD-Ms-gen-answer-2F2} has a fairly transparent structure: it is a superposition of a particular solution to~\eqref{eq:QSD-frac-moments-recurrence} and the general solution to the homogeneous version of~\eqref{eq:QSD-frac-moments-recurrence}. The particular solution is given by the ${}_{2}F_{2}[z]$ term on the right of~\eqref{eq:QSD-Ms-gen-answer-2F2}, and it was already demonstrated explicitly by~\cite{Polunchenko+Pepelyshev:SP2018} that it ``works''. To see that the second term (proportional to $C$) on the right of~\eqref{eq:QSD-Ms-gen-answer-2F2} is, indeed, the general solution to the homogeneous version of~\eqref{eq:QSD-frac-moments-recurrence}, it suffices to iterate the corresponding first-order recurrence forward in $s$, and then simplify the resulting expression using~\eqref{eq:xi-def} and the well-known factorial property of the Gamma function, i.e., $\Gamma(z+1)=z\Gamma(z)$ for any $z\in\mathbb{C}$.
\end{proof}

The question now is how to ``pin down'' the constant $C$ on the right in~\eqref{eq:QSD-Ms-gen-answer-2F2}. Given that~\eqref{eq:QSD-frac-moments-recurrence} is an order-one nonhomogeneous difference equation, the question would be trivial if we had $\mathfrak{M}_{s}$ computed as a function of $A>0$ for at least one particular $s\in\mathbb{R}$. To this end, setting $s=0$ seems the obvious first choice, because $\mathfrak{M}_{0}\equiv1$ for any $A>0$. However, setting $s=0$ in the general solution~\eqref{eq:QSD-Ms-gen-answer-2F2} results in $\mathfrak{M}_{0}=1$ for {\em any} $C$. The reason is because the denominator of the second term on the right of~\eqref{eq:QSD-Ms-gen-answer-2F2} is $\Gamma(-s)$ which has a (simple) pole at $s=0$, so that, no matter what $C$ is, the entire second term is nullified, reducing~\eqref{eq:QSD-Ms-gen-answer-2F2} down to~\eqref{eq:QSD-Mn-answer-2F2}. This circumstance renders the condition $\mathfrak{M}_0\equiv 1$ practically useless for the purposes of finding $C$. As a matter of fact, the same happens for any $s=n\in\mathbb{N}\cup\{0\}$, inasmuch as the Gamma function $\Gamma(z)$ is known to have (simple) poles at all nonnegative-integer values of the argument. This explains why~\cite{Polunchenko+Pepelyshev:SP2018}, in their derivation of solution~\eqref{eq:QSD-Mn-answer-2F2} to recurrence~\eqref{eq:QSD-moments-recurrence}, never had to find the general solution to the homogeneous version of~\eqref{eq:QSD-moments-recurrence}: even if this general solution were found, it would have ended up evaluating to zero for any $n\in\mathbb{N}\cup\{0\}$ anyway, entirely independently of arbitrary constant factor. The ``uselessness'' of the condition $\mathfrak{M}_0\equiv 1$ is easily overcome by the following observation: since
\begin{equation*}
\dfrac{s(s-1)}{2}
+
\lambda
=
0
\;\;
\text{for}
\;\;
s=
\dfrac{1}{2}\pm
\dfrac{\xi(\lambda)}{2},
\end{equation*}
because of~\eqref{eq:xi-def}, we immediately obtain from~\eqref{eq:QSD-frac-moments-recurrence} that
\begin{equation}\label{eq:QSD-M_1over2_pm_xi_over2-answer}
\mathfrak{M}_{-\tfrac{1}{2}\pm\tfrac{\xi(\lambda)}{2}}
=
\dfrac{2\lambda}{1\pm\xi(\lambda)}\,A^{\tfrac{1}{2}\pm\tfrac{\xi(\lambda)}{2}}
=
\dfrac{1\mp\xi(\lambda)}{4}\,A^{\tfrac{1}{2}\pm\tfrac{\xi(\lambda)}{2}},
\end{equation}
where the second equality is, again, due to~\eqref{eq:xi-def}. In other words, some of the $\mathfrak{M}_s$'s are given away ``for free'' by the recurrence~\eqref{eq:QSD-frac-moments-recurrence} itself. Alternatively, from~\eqref{eq:QSD-pdf-answer} and the definition of $\mathfrak{M}_{s}$ we have
\begin{equation}
\begin{split}
\mathfrak{M}_{-\tfrac{1}{2}\pm\tfrac{\xi(\lambda)}{2}}
&\coloneqq
\bigintsss_{0}^{A} x^{-\tfrac{1}{2}\pm\tfrac{\xi(\lambda)}{2}}\,q_{A}(x)\,dx
=
\bigints_{1}^{\infty} \left(\dfrac{t}{A}\right)^{\tfrac{1}{2}\mp\tfrac{\xi(\lambda)}{2}}\,\dfrac{e^{-\tfrac{1}{A}t}\,W_{1,\tfrac{\xi(\lambda)}{2}}\left(\dfrac{2}{A}t\right)}{e^{-\tfrac{1}{A}}\,W_{0,\tfrac{\xi(\lambda)}{2}}\left(\dfrac{2}{A}\right)}\,\dfrac{dt}{t}
=
\\
&
\qquad\qquad
=
\sqrt{\dfrac{A}{2}}\,A^{-\tfrac{1}{2}\pm\tfrac{\xi(\lambda)}{2}}\,\dfrac{W_{\tfrac{1}{2},\tfrac{\xi(\lambda)}{2}\mp\tfrac{1}{2}}\left(\dfrac{2}{A}\right)}{W_{0,\tfrac{\xi(\lambda)}{2}}\left(\dfrac{2}{A}\right)},
\end{split}
\end{equation}
where the last equality is because of~\cite[Identity~2.19.5.6,~p.~217]{Prudnikov+etal:Book1990}, i.e., the definite integral identity
\begin{equation}
\int_{\alpha}^{\infty}(x-\alpha)^{\beta-1}\, x^{\pm b-\tfrac{1}{2}}\,e^{-\tfrac{1}{2}ax}\,W_{\varkappa,b}(ax)\,dx
=
{\alpha}^{\pm b+\tfrac{\beta-1}{2}}\,\Gamma(\beta)\,e^{-\tfrac{\alpha}{2}a}\,W_{\varkappa-\tfrac{\beta}{2},b\pm\tfrac{\beta}{2}}(a),
\end{equation}
valid as long as $\alpha>0$, $\Re(\beta)>0$, and $\Re(a)>0$. Next, from the identity
\begin{equation*}
W_{\varkappa+\tfrac{1}{2},b\pm\tfrac{1}{2}}(z)
=
\dfrac{1\pm2b+z}{2\sqrt{z}}\,W_{\varkappa,b}(z)-\sqrt{z}\left[\dfrac{\partial}{\partial z}W_{\varkappa,b}(z)\right],
\end{equation*}
which is~\cite[Identity~(2.4.22),~p.~25]{Slater:Book1960} and~\cite[Identity~(2.4.23),~p.~25]{Slater:Book1960} combined, we obtain
\begin{equation*}
W_{\tfrac{1}{2},\tfrac{\xi(\lambda)}{2}\pm\tfrac{1}{2}}\left(\dfrac{2}{A}\right)
=
\dfrac{1\pm\xi(\lambda)+2/A}{2\sqrt{2/A}}\,W_{0,\tfrac{\xi(\lambda)}{2}}\left(\dfrac{2}{A}\right)-\sqrt{\dfrac{2}{A}}\left.\left[\dfrac{\partial}{\partial z}W_{0,\tfrac{\xi(\lambda)}{2}}(z)\right]\right|_{z=\tfrac{2}{A}},
\end{equation*}
whence
\begin{equation*}
W_{\tfrac{1}{2},\tfrac{\xi(\lambda)}{2}\pm\tfrac{1}{2}}\left(\dfrac{2}{A}\right)
=
\dfrac{1\mp\xi(\lambda)}{2}\,\sqrt{\dfrac{A}{2}}\,W_{0,\tfrac{\xi(\lambda)}{2}}\left(\dfrac{2}{A}\right),
\end{equation*}
because
\begin{equation*}
\left.\left[\dfrac{\partial}{\partial z}W_{0,\tfrac{\xi(\lambda)}{2}}(z)\right]\right|_{z=\tfrac{2}{A}}
=
\dfrac{1}{2}\,W_{0,\tfrac{\xi(\lambda)}{2}}\left(\dfrac{2}{A}\right),
\end{equation*}
as is implied by~\cite[Identity~(2.4.24),~p.~25]{Slater:Book1960}, i.e., the identity
\begin{equation*}
W_{\varkappa+1,b}(z)
=
\left(\dfrac{z}{2}-\varkappa\right)W_{\varkappa,b}(z)-z\left[\dfrac{\partial}{\partial z}W_{\varkappa,b}(z)\right],
\end{equation*}
and condition~\eqref{eq:lambda-eqn}. Now, putting all of the above together~\eqref{eq:QSD-M_1over2_pm_xi_over2-answer} follows.

To be able use~\eqref{eq:QSD-Ms-gen-answer-2F2} and~\eqref{eq:QSD-M_1over2_pm_xi_over2-answer} to get $C$ in as simple a form as possible, it is necessary to first evaluate~\eqref{eq:QSD-Ms-gen-answer-2F2} for $s=-1/2\pm\xi(\lambda)/2$, and ``massage'' the result so as to bring it to a suitable form. One way to achieve this is through the identity
\begin{equation*}
z(b-a)
{}_{2}F_{2}
\left[
    \setlength{\arraycolsep}{0pt}
    \setlength{\extrarowheight}{2pt}
    \begin{array}{@{} c@{{}{}} @{}}
    a+1,b+1
    \\[1ex]
    c+1,d+1
    \\[5pt]
    \end{array}
    \;\middle|\;
    z
\right]
=
cd
\left(
{}_{2}F_{2}
\left[
    \setlength{\arraycolsep}{0pt}
    \setlength{\extrarowheight}{2pt}
    \begin{array}{@{} c@{{}{}} @{}}
    a+1,b
    \\[1ex]
    c,d
    \\[5pt]
    \end{array}
    \;\middle|\;
    z
\right]
-
{}_{2}F_{2}
\left[
    \setlength{\arraycolsep}{0pt}
    \setlength{\extrarowheight}{2pt}
    \begin{array}{@{} c@{{}{}} @{}}
    a,b+1
    \\[1ex]
    c,d
    \\[5pt]
    \end{array}
    \;\middle|\;
    z
\right]
\right),
\end{equation*}
which is one of the contiguous relations that the ${}_{2}F_{2}[z]$ function is known to satisfy. Since from~\eqref{eq:pFq-function-def} it is easy to see that
\begin{equation*}
{}_{2}F_{2}
\left[
    \setlength{\arraycolsep}{0pt}
    \setlength{\extrarowheight}{2pt}
    \begin{array}{@{} c@{{}{}} @{}}
    0,a_2
    \\[1ex]
    b_1,b_2
    \\[2pt]
    \end{array}
    \;\middle|\;
    z
\right]
=
1
\;\;
\text{and}
\;\;
{}_{2}F_{2}
\left[
    \setlength{\arraycolsep}{0pt}
    \setlength{\extrarowheight}{2pt}
    \begin{array}{@{} c@{{}{}} @{}}
    a_1,a_2
    \\[1ex]
    a_1,b_2
    \\[2pt]
    \end{array}
    \;\middle|\;
    z
\right]
=
{}_{1}F_{1}
\left[
    \setlength{\arraycolsep}{0pt}
    \setlength{\extrarowheight}{2pt}
    \begin{array}{@{} c@{{}{}} @{}}
    a_2
    \\[1ex]
    b_2
    \\[2pt]
    \end{array}
    \;\middle|\;
    z
\right]
\;
\text{for any appropriate $a_1$, $a_2$, $b_1$ and $b_2$},
\end{equation*}
the above contiguous relation leads to
\begin{equation*}
\begin{split}
{}_{2}F_{2}
\left[
    \setlength{\arraycolsep}{0pt}
    \setlength{\extrarowheight}{2pt}
    \begin{array}{@{} c@{{}{}} @{}}
    1,\dfrac{1}{2}\pm\dfrac{\xi(\lambda)}{2}
    \\[1ex]
    2\pm\xi(\lambda),2
    \\[5pt]
    \end{array}
    \;\middle|\;
    \dfrac{2}{A}
\right]
&=
A\,
\dfrac{\xi(\lambda)\pm 1}{\xi(\lambda)\mp 1}
\left(
{}_{1}F_{1}
\left[
    \setlength{\arraycolsep}{0pt}
    \setlength{\extrarowheight}{2pt}
    \begin{array}{@{} c@{{}{}} @{}}
    -\dfrac{1}{2}\pm\dfrac{\xi(\lambda)}{2}
    \\[1ex]
    1\pm\xi(\lambda)
    \\[5pt]
    \end{array}
    \;\middle|\;
    \dfrac{2}{A}
\right]
-
1
\right),
\end{split}
\end{equation*}
so that subsequently from~\eqref{eq:QSD-Ms-answer-2F2} we find
\begin{equation*}
\begin{split}
\mathfrak{M}_{-\tfrac{1}{2}\pm\tfrac{\xi(\lambda)}{2}}
&=
\dfrac{1\mp \xi(\lambda)}{4}\,A^{\tfrac{1}{2}\pm\tfrac{\xi(\lambda)}{2}}
\left(
1
-
{}_{1}F_{1}
\left[
    \setlength{\arraycolsep}{0pt}
    \setlength{\extrarowheight}{2pt}
    \begin{array}{@{} c@{{}{}} @{}}
    -\dfrac{1}{2}\mp\dfrac{\xi(\lambda)}{2}
    \\[1ex]
    1\mp\xi(\lambda)
    \\[5pt]
    \end{array}
    \;\middle|\;
    \dfrac{2}{A}
\right]
\right)
+
\\
&\qquad\qquad\qquad+
C\,2^{-\tfrac{1}{2}\pm\tfrac{\xi(\lambda)}{2}}\,\Gamma(1\mp\xi(\lambda))\left/\Gamma\left(\dfrac{1}{2}\mp \dfrac{\xi(\lambda)}{2}\right)\right.
,
\end{split}
\end{equation*}
whence, in view of~\eqref{eq:QSD-M_1over2_pm_xi_over2-answer}, we finally obtain
\begin{equation}\label{eq:QSD-Ms-C-answer}
C
\equiv
C_{A}
=
\dfrac{1\mp\xi(\lambda)}{2\,\Gamma(1\mp\xi(\lambda))}\left(\dfrac{A}{2}\right)^{\tfrac{1}{2}\pm\tfrac{\xi(\lambda)}{2}}
\Gamma\left(\dfrac{1}{2}\mp\dfrac{\xi(\lambda)}{2}\right)
{}_{1}F_{1}
\left[
    \setlength{\arraycolsep}{0pt}
    \setlength{\extrarowheight}{2pt}
    \begin{array}{@{} c@{{}{}} @{}}
    -\dfrac{1}{2}\mp\dfrac{\xi(\lambda)}{2}
    \\[1ex]
    1\mp\xi(\lambda)
    \\[5pt]
    \end{array}
    \;\middle|\;
    \dfrac{2}{A}
\right]
,
\end{equation}
where again $\lambda\equiv\lambda_{A}>0$ is determined by~\eqref{eq:lambda-eqn} and $\xi(\lambda)$ is defined in~\eqref{eq:xi-def}.

We have obtained not one but {\em two} expressions for $C$, although the difference between the expressions is only in the sign of $\xi(\lambda)$. The two expressions are equivalent, and, in a nutshell, it has to do with condition~\eqref{eq:lambda-eqn} and the symmetry of the Whittaker $W$ function with respect to its second index, i.e., $W_{a,b}(z)=W_{a,-b}(z)$. To give a more detailed explanation we turn to~\cite[Identity~13.1.34,~p.~505]{Abramowitz+Stegun:Handbook1964}, i.e., the identity
\begin{equation}\label{eq:WhitWM-identity}
W_{\varkappa,b}(z)
=
\dfrac{\Gamma(-2b)}{\Gamma(1/2-b-\varkappa)}M_{\varkappa,b}(z)+\dfrac{\Gamma(2b)}{\Gamma(1/2+b-\varkappa)}M_{\varkappa,-b}(z),
\end{equation}
where $M_{a,b}(z)$ denotes the Whittaker $M$ function. On account of~\eqref{eq:lambda-eqn} this identity gives
\begin{equation*}
\dfrac{\Gamma(-\xi(\lambda))}{\Gamma\left(-\dfrac{1}{2}-\dfrac{\xi(\lambda)}{2}\right)}M_{1,\tfrac{\xi(\lambda)}{2}}\left(\dfrac{2}{A}\right)
=
-\dfrac{\Gamma(\xi(\lambda))}{\Gamma\left(-\dfrac{1}{2}+\dfrac{\xi(\lambda)}{2}\right)}M_{1,-\tfrac{\xi(\lambda)}{2}}\left(\dfrac{2}{A}\right),
\end{equation*}
or equivalently
\begin{equation*}
\Gamma\left(\dfrac{1}{2}+\dfrac{\xi(\lambda)}{2}\right)\dfrac{1+\xi(\lambda)}{\Gamma(1+\xi(\lambda))}\,M_{1,\tfrac{\xi(\lambda)}{2}}\left(\dfrac{2}{A}\right)
=
\Gamma\left(\dfrac{1}{2}-\dfrac{\xi(\lambda)}{2}\right)\dfrac{1-\xi(\lambda)}{\Gamma(1-\xi(\lambda))}\,M_{1,-\tfrac{\xi(\lambda)}{2}}\left(\dfrac{2}{A}\right),
\end{equation*}
because, again, $\Gamma(z+1)=z\,\Gamma(z)$ for any $z\in\mathbb{C}$. Next, since the Whittaker $M$ function and the Kummer ${}_{1}F_{1}[z]$ function are related, and the relationship is given by~\cite[Identities~(3a) and~(4),~p.~11]{Buchholz:Book1969}, which together amount to
\begin{equation*}
M_{\varkappa,\pm b}(z)
=
z^{\tfrac{1}{2}\pm b}\,e^{-\tfrac{z}{2}}\,
{}_{1}F_{1}
\left[
    \setlength{\arraycolsep}{0pt}
    \setlength{\extrarowheight}{2pt}
    \begin{array}{@{} c@{{}{}} @{}}
    \dfrac{1}{2}\pm b-\varkappa
    \\[1ex]
    1\pm 2b
    \\[5pt]
    \end{array}
    \;\middle|\;
    z
\right]
,
\end{equation*}
we further obtain
\begin{equation*}
\begin{split}
&
\Gamma\left(\dfrac{1}{2}+\dfrac{\xi(\lambda)}{2}\right)\dfrac{1+\xi(\lambda)}{\Gamma(1+\xi(\lambda))}
\left(\dfrac{2}{A}\right)^{\tfrac{\xi(\lambda)}{2}}
{}_{1}F_{1}
\left[
    \setlength{\arraycolsep}{0pt}
    \setlength{\extrarowheight}{2pt}
    \begin{array}{@{} c@{{}{}} @{}}
    -\dfrac{1}{2}+\dfrac{\xi(\lambda)}{2}
    \\[1ex]
    1+\xi(\lambda)
    \\[5pt]
    \end{array}
    \;\middle|\;
    \dfrac{2}{A}
\right]
=\\
&\qquad\qquad
=
\Gamma\left(\dfrac{1}{2}-\dfrac{\xi(\lambda)}{2}\right)\dfrac{1-\xi(\lambda)}{\Gamma(1-\xi(\lambda))}
\left(\dfrac{2}{A}\right)^{-\tfrac{\xi(\lambda)}{2}}
{}_{1}F_{1}
\left[
    \setlength{\arraycolsep}{0pt}
    \setlength{\extrarowheight}{2pt}
    \begin{array}{@{} c@{{}{}} @{}}
    -\dfrac{1}{2}-\dfrac{\xi(\lambda)}{2}
    \\[1ex]
    1-\xi(\lambda)
    \\[5pt]
    \end{array}
    \;\middle|\;
    \dfrac{2}{A}
\right],
\end{split}
\end{equation*}
and it is now evident that the sign of $\xi(\lambda)$ on the right of~\eqref{eq:QSD-Ms-C-answer} is, indeed, irrelevant.

It is possible to simplify~\eqref{eq:QSD-Ms-C-answer} considerably, with the help of the identity
\begin{equation*}
\lambda\,A\,\Gamma\left(\dfrac{\xi(\lambda)-1}{2}\right)W_{0,\tfrac{\xi(\lambda)}{2}}\left(\dfrac{2}{A}\right)M_{1,\tfrac{\xi(\lambda)}{2}}\left(\dfrac{2}{A}\right)
=
-\Gamma(\xi(\lambda)+1),
\end{equation*}
which was established by~\cite[p.~1373]{Polunchenko+Pepelyshev:SP2018}. In view of this identity and the earlier discussion one can see that
\begin{equation*}
\dfrac{1}{e^{-\tfrac{1}{A}}\,W_{0,\tfrac{\xi(\lambda)}{2}}\left(\dfrac{2}{A}\right)}
=
\Gamma\left(\dfrac{1}{2}+\dfrac{\xi(\lambda)}{2}\right)\dfrac{1+\xi(\lambda)}{2\,\Gamma(1+\xi(\lambda))}
\left(\dfrac{2}{A}\right)^{-\tfrac{1}{2}+\tfrac{\xi(\lambda)}{2}}
{}_{1}F_{1}
\left[
    \setlength{\arraycolsep}{0pt}
    \setlength{\extrarowheight}{2pt}
    \begin{array}{@{} c@{{}{}} @{}}
    -\dfrac{1}{2}+\dfrac{\xi(\lambda)}{2}
    \\[1ex]
    1+\xi(\lambda)
    \\[5pt]
    \end{array}
    \;\middle|\;
    \dfrac{2}{A}
\right]
,
\end{equation*}
whence
\begin{equation*}
C
=
1\left/\left[e^{-\tfrac{1}{A}}\,W_{0,\tfrac{\xi(\lambda)}{2}}\left(\dfrac{2}{A}\right)\right]\right.,
\end{equation*}
i.e., $C$ is precisely the normalizing factor in formulae~\eqref{eq:QSD-pdf-answer}--\eqref{eq:QSD-cdf-answer} for the quasi-stationary pdf $q_{A}(x)$ and cdf $Q_{A}(x)$. Other than being far more compact than~\eqref{eq:QSD-Ms-C-answer}, the new formula for $C$ also makes it clear that the latter is invariant with respect to the sign of $\xi(\lambda)$: the reason is the symmetry of the Whittaker $W$ function with respect to its second index, i.e., $W_{a,b}(z)=W_{a,-b}(z)$. Moreover, since
\begin{equation*}
\lim_{A\to+\infty}\left[e^{-\tfrac{1}{A}}\,W_{0,\tfrac{\xi(\lambda_{A})}{2}}\left(\dfrac{2}{A}\right)\right]
=
1,
\end{equation*}
as was previously shown by~\cite[p.~139]{Polunchenko:SA2016}, it can be readily concluded that $\lim_{A\to+\infty}C_{A}=1$.

We are now ready to state our main result.
\begin{lemma}
For every $A>0$ fixed, the fractional moment $\mathfrak{M}_{s}$ of order $s\in\mathbb{R}$ of the quasi-stationary distribution~\eqref{eq:QSD-pdf-answer}--\eqref{eq:QSD-cdf-answer} is given by
\begin{equation}\label{eq:QSD-Ms-answer-2F2}
\begin{split}
\mathfrak{M}_{s}
\equiv
\mathfrak{M}_{s}(A)
&=
\dfrac{2\lambda\, A^{s}}{s(s-1)+2\lambda}
\,
{}_{2}F_{2}
\left[
    \setlength{\arraycolsep}{0pt}
    \setlength{\extrarowheight}{2pt}
    \begin{array}{@{} c@{{}{}} @{}}
    1,-s
    \\[1ex]
    \dfrac{3}{2}+\dfrac{\xi(\lambda)}{2}-s,\dfrac{3}{2}-\dfrac{\xi(\lambda)}{2}-s
    \\[5pt]
    \end{array}
    \;\middle|\;
    \dfrac{2}{A}
\right]
+
\\
&\qquad\qquad+
C\,\dfrac{2^{s}}{\Gamma(-s)}\,\Gamma\left(\dfrac{1}{2}+\dfrac{\xi(\lambda)}{2}-s\right)\Gamma\left(\dfrac{1}{2}-\dfrac{\xi(\lambda)}{2}-s\right)
,
\;
s\in\mathbb{R},
\end{split}
\end{equation}
where
\begin{equation}\label{eq:QSD-Ms-answer-C-2F2}
\begin{split}
C
\equiv
C_A
&=
\dfrac{1\pm\xi(\lambda)}{2\,\Gamma(1\pm\xi(\lambda))}\left(\dfrac{A}{2}\right)^{\tfrac{1}{2}\mp\tfrac{\xi(\lambda)}{2}}
\Gamma\left(\dfrac{1}{2}\pm\dfrac{\xi(\lambda)}{2}\right)
{}_{1}F_{1}
\left[
    \setlength{\arraycolsep}{0pt}
    \setlength{\extrarowheight}{2pt}
    \begin{array}{@{} c@{{}{}} @{}}
    -\dfrac{1}{2}\pm\dfrac{\xi(\lambda)}{2}
    \\[1ex]
    1\pm\xi(\lambda)
    \\[5pt]
    \end{array}
    \;\middle|\;
    \dfrac{2}{A}
\right]
=
\\
&
=
1\left/\left[e^{-\tfrac{1}{A}}\,W_{0,\pm\tfrac{\xi(\lambda)}{2}}\left(\dfrac{2}{A}\right)\right]\right.
\end{split}
\end{equation}
with $\lambda\equiv\lambda_A\;(>0)$ determined by~\eqref{eq:lambda-eqn} and $\xi(\lambda)$ defined in~\eqref{eq:xi-def}; recall also that ${}_{2}F_{2}[z]$ denotes the corresponding special case of the generalized hypergeometric function~\eqref{eq:pFq-function-def}.
\end{lemma}
The obtained $s$-th order fractional moment formula~\eqref{eq:QSD-Ms-answer-2F2} is valid even if $s=1/2\pm\xi(\lambda)/2+k$ for $k\in\mathbb{N}\cup\{0\}$, so long as it is used with care: the first term on the right of~\eqref{eq:QSD-Ms-answer-2F2} experiences a singularity at $s=1/2\pm\xi(\lambda)/2+k$, and one of the two Gamma functions in the second term on the right of~\eqref{eq:QSD-Ms-answer-2F2} experiences a (simple) pole whenever $s=1/2\pm\xi(\lambda)/2+k$. To better understand how to use formula~\eqref{eq:QSD-Ms-answer-2F2} for $s=1/2\pm\xi(\lambda)/2+k$, suppose, for simplicity, that $k=0$, and consider $s=s_{\varepsilon}=1/2\pm\xi(\lambda)/2+\varepsilon$ for some sufficiently small $\varepsilon\neq 0$. For this choice of $s$ from~\eqref{eq:QSD-Ms-answer-2F2} and~\eqref{eq:QSD-Ms-answer-C-2F2} and repeated use of $\Gamma(1+z)=z\,\Gamma(z)$, $z\in\mathbb{C}$, we have
\begin{equation*}
\begin{split}
&\mathfrak{M}_{\tfrac{1}{2}\pm\tfrac{\xi(\lambda)}{2}+\varepsilon}
=
2\lambda\, A^{\tfrac{1}{2}\pm\tfrac{\xi(\lambda)}{2}}\Bigggg\{\dfrac{A^{\varepsilon}}{\varepsilon(\varepsilon\pm\xi(\lambda))}\,
{}_{2}F_{2}
\left[
    \setlength{\arraycolsep}{0pt}
    \setlength{\extrarowheight}{2pt}
    \begin{array}{@{} c@{{}{}} @{}}
    1,-\dfrac{1}{2}\mp\dfrac{\xi(\lambda)}{2}-\varepsilon
    \\[1ex]
    1-\varepsilon,1\mp\xi(\lambda)-\varepsilon
    \\[5pt]
    \end{array}
    \;\middle|\;
    \dfrac{2}{A}
\right]
+
\\
&
\quad+
\dfrac{2^{\varepsilon}\,\Gamma(1-\varepsilon)\,\Gamma(\mp\xi(\lambda)-\varepsilon)}{\varepsilon\,\Gamma(1\mp\xi(\lambda))}\,\Gamma\left(-\dfrac{1}{2}\mp\dfrac{\xi(\lambda)}{2}\right)
\left.
{}_{1}F_{1}
\left[
    \setlength{\arraycolsep}{0pt}
    \setlength{\extrarowheight}{2pt}
    \begin{array}{@{} c@{{}{}} @{}}
    -\dfrac{1}{2}\mp\dfrac{\xi(\lambda)}{2}
    \\[1ex]
    1\mp\xi(\lambda)
    \\[5pt]
    \end{array}
    \;\middle|\;
    \dfrac{2}{A}
\right]
\right/
\Gamma\left(-\dfrac{1}{2}\mp\dfrac{\xi(\lambda)}{2}-\varepsilon\right)
\Bigggg\}
,
\end{split}
\end{equation*}
so that each of the two terms inside the braces can be seen to have a singularity of type $1/\varepsilon$ as $\varepsilon\to0$. However, these singularities ``undo'' each other in such a way that
\begin{equation*}
\mathfrak{M}_{\tfrac{1}{2}\pm\tfrac{\xi(\lambda)}{2}}
=
\lim_{\varepsilon\to0}\mathfrak{M}_{\tfrac{1}{2}\pm\tfrac{\xi(\lambda)}{2}+\varepsilon}
\end{equation*}
ends up being well-defined. To show this and properly handle the singularities observe first that
\begin{equation*}
\begin{split}
&\mathfrak{M}_{\tfrac{1}{2}\pm\tfrac{\xi(\lambda)}{2}+\varepsilon}
=
\dfrac{2\lambda}{\varepsilon}\, A^{\tfrac{1}{2}\pm\tfrac{\xi(\lambda)}{2}}\,2^{\varepsilon}\left[\Gamma(-\varepsilon\mp\xi(\lambda))\,\Gamma(1-\varepsilon)\left/
\Gamma\left(-\dfrac{1}{2}\mp\dfrac{\xi(\lambda)}{2}-\varepsilon\right)\right.\right]
\times
\\
&\qquad\qquad
\times\Bigggg\{\dfrac{1}{\Gamma(1\mp\xi(\lambda))}\,\Gamma\left(-\dfrac{1}{2}\mp\dfrac{\xi(\lambda)}{2}\right)
{}_{1}F_{1}
\left[
    \setlength{\arraycolsep}{0pt}
    \setlength{\extrarowheight}{2pt}
    \begin{array}{@{} c@{{}{}} @{}}
    -\dfrac{1}{2}\mp\dfrac{\xi(\lambda)}{2}
    \\[1ex]
    1\mp\xi(\lambda)
    \\[5pt]
    \end{array}
    \;\middle|\;
    \dfrac{2}{A}
\right]
-
\\
&\quad
-
\dfrac{\left(A/2\right)^{\varepsilon}}{\Gamma(1-\varepsilon)\,\Gamma(1\mp\xi(\lambda)-\varepsilon)}\,
\Gamma\left(-\dfrac{1}{2}\mp\dfrac{\xi(\lambda)}{2}-\varepsilon\right)
\,
{}_{2}F_{2}
\left[
    \setlength{\arraycolsep}{0pt}
    \setlength{\extrarowheight}{2pt}
    \begin{array}{@{} c@{{}{}} @{}}
    1,-\dfrac{1}{2}\mp\dfrac{\xi(\lambda)}{2}-\varepsilon
    \\[1ex]
    1-\varepsilon,1\mp\xi(\lambda)-\varepsilon
    \\[5pt]
    \end{array}
    \;\middle|\;
    \dfrac{2}{A}
\right]
\Bigggg\}
,
\end{split}
\end{equation*}
whence on account of
\begin{equation*}
\begin{split}
&
\dfrac{\left(A/2\right)^{\varepsilon}}{\Gamma(1-\varepsilon)\,\Gamma(1\mp\xi(\lambda)-\varepsilon)}\,
\Gamma\left(-\dfrac{1}{2}\mp\dfrac{\xi(\lambda)}{2}-\varepsilon\right)
\,
{}_{2}F_{2}
\left[
    \setlength{\arraycolsep}{0pt}
    \setlength{\extrarowheight}{2pt}
    \begin{array}{@{} c@{{}{}} @{}}
    1,-\dfrac{1}{2}\mp\dfrac{\xi(\lambda)}{2}-\varepsilon
    \\[1ex]
    1-\varepsilon,1\mp\xi(\lambda)-\varepsilon
    \\[5pt]
    \end{array}
    \;\middle|\;
    \dfrac{2}{A}
\right]
=
\\
&\qquad\qquad
=
\dfrac{1}{\Gamma(1\mp\xi(\lambda))}\,\Gamma\left(-\dfrac{1}{2}\mp\dfrac{\xi(\lambda)}{2}\right)
{}_{1}F_{1}
\left[
    \setlength{\arraycolsep}{0pt}
    \setlength{\extrarowheight}{2pt}
    \begin{array}{@{} c@{{}{}} @{}}
    -\dfrac{1}{2}\mp\dfrac{\xi(\lambda)}{2}
    \\[1ex]
    1\mp\xi(\lambda)
    \\[5pt]
    \end{array}
    \;\middle|\;
    \dfrac{2}{A}
\right]
+
\\
&
+
\dfrac{\partial }{\partial \delta}\left.\Bigggg\{
\dfrac{\left(A/2\right)^{\delta}}{\Gamma(1-\delta)\,\Gamma(1\mp\xi(\lambda)-\delta)}\,
\Gamma\left(-\dfrac{1}{2}\mp\dfrac{\xi(\lambda)}{2}-\delta\right)
\,
{}_{2}F_{2}
\left[
    \setlength{\arraycolsep}{0pt}
    \setlength{\extrarowheight}{2pt}
    \begin{array}{@{} c@{{}{}} @{}}
    1,-\dfrac{1}{2}\mp\dfrac{\xi(\lambda)}{2}-\delta
    \\[1ex]
    1-\delta,1\mp\xi(\lambda)-\delta
    \\[5pt]
    \end{array}
    \;\middle|\;
    \dfrac{2}{A}
\right]
\Bigggg\}\right|_{\delta=\delta(\varepsilon)}
\varepsilon,
\end{split}
\end{equation*}
and recalling~\eqref{eq:PochhammerGamma-def} and~\eqref{eq:pFq-function-def} one can conclude that
\begin{equation}\label{eq:M_1_over_2_pm_xi_over_2-answer}
\begin{split}
&\mathfrak{M}_{\tfrac{1}{2}\pm\tfrac{\xi(\lambda)}{2}}
\equiv
\mathfrak{M}_{\tfrac{1}{2}\pm\tfrac{\xi(\lambda)}{2}}(A)
=
\lim_{\varepsilon\to0}\mathfrak{M}_{\tfrac{1}{2}\pm\tfrac{\xi(\lambda)}{2}+\varepsilon}
=
-2\lambda\, A^{\tfrac{1}{2}\pm\tfrac{\xi(\lambda)}{2}}\left[\Gamma(\mp\xi(\lambda))\left/
\Gamma\left(-\dfrac{1}{2}\mp\dfrac{\xi(\lambda)}{2}\right)\right.\right]
\times
\\
&\quad
\times
\dfrac{\partial }{\partial \delta}\left.\Bigggg\{
\dfrac{\left(A/2\right)^{\delta}}{\Gamma(1-\delta)\,\Gamma(1\mp\xi(\lambda)-\delta)}\,
\Gamma\left(-\dfrac{1}{2}\mp\dfrac{\xi(\lambda)}{2}-\delta\right)
\,
{}_{2}F_{2}
\left[
    \setlength{\arraycolsep}{0pt}
    \setlength{\extrarowheight}{2pt}
    \begin{array}{@{} c@{{}{}} @{}}
    1,-\dfrac{1}{2}\mp\dfrac{\xi(\lambda)}{2}-\delta
    \\[1ex]
    1-\delta,1\mp\xi(\lambda)-\delta
    \\[5pt]
    \end{array}
    \;\middle|\;
    \dfrac{2}{A}
\right]
\Bigggg\}\right|_{\delta=0}
=
\\
&\quad
=
-2\lambda\, A^{\tfrac{1}{2}\pm\tfrac{\xi(\lambda)}{2}}\left[\Gamma(\mp\xi(\lambda))\left/
\Gamma\left(-\dfrac{1}{2}\mp\dfrac{\xi(\lambda)}{2}\right)\right.\right]
\times
\\
&\qquad\qquad
\times
\sum_{j=0}^{\infty}
\Biggg\{
\Gamma\left(-\dfrac{1}{2}\mp\dfrac{\xi(\lambda)}{2}+j\right)
\dfrac{(2/A)^j}{j!\,\Gamma(1\mp\xi(\lambda))}\Bigg[\psi(1+j)+\psi(1\mp\xi(\lambda)+j)
-
\\
&\qquad\qquad\qquad\qquad
-
\psi\left(-\dfrac{1}{2}\mp\dfrac{\xi(\lambda)}{2}+j\right)-\log\left(\dfrac{2}{A}\right)\Bigg]
\Biggg\}
=
\\
&\quad
=
\pm\dfrac{2\lambda}{\xi(\lambda)}\, A^{\tfrac{1}{2}\pm\tfrac{\xi(\lambda)}{2}}
\sum_{j=0}^{\infty}
\Biggg\{
\left(-\dfrac{1}{2}\mp\dfrac{\xi(\lambda)}{2}\right)_{j}
\dfrac{(2/A)^j}{j!\,(1\mp\xi(\lambda))_{j}}\Bigg[\psi(1+j)+\psi(1\mp\xi(\lambda)+j)
-
\\
&\qquad\qquad\qquad\qquad
-
\psi\left(-\dfrac{1}{2}\mp\dfrac{\xi(\lambda)}{2}+j\right)-\log\left(\dfrac{2}{A}\right)\Bigg]
\Biggg\},
\end{split}
\end{equation}
where $\psi(z)$ denotes the digamma function~\eqref{eq:digamma-fun-def}; recall again that $\lambda\equiv\lambda_{A}>0$ is determined by equation~\eqref{eq:lambda-eqn}.

The case of $s=1/2\pm\xi(\lambda)/2+k$, $k\in\mathbb{N}$, can be treated in a similar manner, leading to
\begin{equation*}
\begin{split}
&\mathfrak{M}_{\tfrac{1}{2}\pm\tfrac{\xi(\lambda)}{2}+k}
\equiv
\mathfrak{M}_{\tfrac{1}{2}\pm\tfrac{\xi(\lambda)}{2}+k}(A)
=
\dfrac{(-2)^{k}}{k!\,\Gamma(1\pm\xi(\lambda)+k)}\,\Gamma\left(\dfrac{3}{2}\pm\dfrac{\xi(\lambda)}{2}+k\right)\,\times
\\
&\qquad\qquad
\times\Biggg\{
\Gamma(1\pm\xi(\lambda))\,\mathfrak{M}_{\tfrac{1}{2}\pm\tfrac{\xi(\lambda)}{2}}(A)\left/\Gamma\left(\dfrac{3}{2}\pm\dfrac{\xi(\lambda)}{2}\right)\right.
-
\\
&\qquad\qquad\qquad
-
\lambda\,A^{\tfrac{3}{2}\pm\tfrac{\xi(\lambda)}{2}}\,\sum_{j=0}^{k-1}\Bigg[\left(-\dfrac{A}{2}\right)^{j}\,j!\,\Gamma(1\pm\xi(\lambda)+j)\left/\Gamma\left(\dfrac{5}{2}\pm\dfrac{\xi(\lambda)}{2}+j\right)\right.\Bigg]
\Biggg\}
=
\\
&\qquad
=
\dfrac{(-2)^{k}}{k!\,(1\pm\xi(\lambda))_{k}}\,\left(\dfrac{3}{2}\pm\dfrac{\xi(\lambda)}{2}\right)_{k}\,\Biggg\{
\mathfrak{M}_{\tfrac{1}{2}\pm\tfrac{\xi(\lambda)}{2}}(A)
-
\\
&\qquad\qquad\qquad
-
\dfrac{2\lambda}{3\pm\xi(\lambda)}\,A^{\tfrac{3}{2}\pm\tfrac{\xi(\lambda)}{2}}\,\sum_{j=0}^{k-1}\Bigg[\left(-\dfrac{A}{2}\right)^{j}\,j!\,(1\pm\xi(\lambda))_{j}\left/\left(\dfrac{5}{2}\pm\dfrac{\xi(\lambda)}{2}\right)_{j}\right.\Bigg]
\Biggg\}
,
\;\;
k\in\mathbb{N},
\end{split}
\end{equation*}
where $\mathfrak{M}_{\tfrac{1}{2}\pm\tfrac{\xi(\lambda)}{2}}\equiv\mathfrak{M}_{\tfrac{1}{2}\pm\tfrac{\xi(\lambda)}{2}}(A)$ is as in~\eqref{eq:M_1_over_2_pm_xi_over_2-answer}.

\section{Concluding remarks}
\label{sec:conclusion}
It was pointed out in the introduction that the quasi-stationary distribution~\eqref{eq:QSD-def} of the Shiryaev process~\eqref{eq:Rt_r-def} converges (weakly) to the process' stationary distribution~\eqref{eq:SR-StDist-def}, as $A\to+\infty$; see \cite{Pollak+Siegmund:B85,Pollak+Siegmund:JAP1996} and~\cite{Li+Polunchenko:SA2019}. Since the stationary distribution---explicitly given by~\eqref{eq:SR-StDist-answer}---has fractional moments of orders $s<1$ only, it stands to reason that, unless $s<1$, the $s$-th fractional moment of the quasi-stationary distribution~\eqref{eq:QSD-pdf-answer}--\eqref{eq:QSD-cdf-answer} is divergent as $A\to+\infty$. This can be verified by computing the corresponding limit explicitly through the $s$-th fractional moment formula~\eqref{eq:QSD-Ms-answer-2F2}. Specifically, recalling~\eqref{eq:lambda-dbl-ineq-asymp} whereby $\lim_{A\to+\infty}\big(A\,\lambda_{A}^{1+\delta}\big)=0$ for any $\delta>0$ but $\lim_{A\to+\infty}\lambda_A=0$, so that $\lim_{A\to+\infty}\xi(\lambda_{A})=1$ owing to~\eqref{eq:xi-def}, we find
\begin{equation*}
\lim_{A\to+\infty}
\left\{
{}_{2}F_{2}
\left[
    \setlength{\arraycolsep}{0pt}
    \setlength{\extrarowheight}{2pt}
    \begin{array}{@{} c@{{}{}} @{}}
    1,-s
    \\[1ex]
    \dfrac{3}{2}+\dfrac{\xi(\lambda)}{2}-s,\dfrac{3}{2}-\dfrac{\xi(\lambda)}{2}-s
    \\[5pt]
    \end{array}
    \;\middle|\;
    \dfrac{2}{A}
\right]
\right\}
=
{}_{2}F_{2}
\left[
    \setlength{\arraycolsep}{0pt}
    \setlength{\extrarowheight}{2pt}
    \begin{array}{@{} c@{{}{}} @{}}
    1,-s
    \\[1ex]
    2-s,1-s
    \\[5pt]
    \end{array}
    \;\middle|\;
    0
\right]
=
1,
\;
\text{for any}
\;
s<1,
\end{equation*}
where the second equality follows directly from the definition~\eqref{eq:pFq-function-def} of the generalized hypergeometric function. Hence, from the $s$-th fractional moment formula~\eqref{eq:QSD-Ms-answer-2F2} we obtain
\begin{equation}
\lim_{A\to+\infty}\mathfrak{M}_{s}(A)
=
2^{s}\,\Gamma(1-s),
\;
\text{for any}
\;
s<1,
\end{equation}
where we also used $\lim_{A\to+\infty}C_{A}=1$ mentioned above. The obtained result is precisely the fractional moment of order $s<1$ of the stationary distribution~\eqref{eq:SR-StDist-answer}, because of~\cite[Identity~3.381.4,~p.~348]{Gradshteyn+Ryzhik:Book2014} which states that
\begin{equation*}
\bigintsss_{0}^{+\infty}
x^{a-1}\,e^{-bx}\,dx
=
\dfrac{\Gamma(a)}{b^{a}},
\end{equation*}
provided $\Re(a)>0$ and $\Re(b)>0$.

As a final comment, we note that the $s$-th fractional moment formula~\eqref{eq:QSD-Ms-answer-2F2} can also be used to compute the $\log$-moment of the quasi-stationary distribution. Specifically, if $X$ is, again, a random variable having the quasi-stationary distribution $q_{A}(x)$ given by~\eqref{eq:QSD-pdf-answer}, then
\begin{equation*}
\EV\big[\log(X)\big]
=
\log(A)-\dfrac{1}{\lambda_{A}}\left(\mathfrak{M}_{-1}(A)-\dfrac{1}{2}\right),
\end{equation*}
where $\mathfrak{M}_{-1}(A)$ is as in~\eqref{eq:QSD-Ms-answer-2F2} and $\lambda_{A}$ is determined by~\eqref{eq:lambda-eqn}.

\section*{Acknowledgements}
The authors are thankful to the two anonymous referees whose constructive feedback helped improve the quality of the manuscript.

The effort of A.S.~Polunchenko was supported, in part, by the Simons Foundation via a Collaboration Grant in Mathematics under Award \#\,304574.

The work of A. Pepelyshev was partially supported by the Russian Foundation for Basic Research under Project \#\,17-01-00161.
%

\singlespacing
\bibliographystyle{abbrvnat}      
\bibliography{main,physics,special-functions,finance,stochastic-processes,differential-equations}

\end{document}